\documentclass{llncs}

\usepackage{makeidx}  
\usepackage{amsmath}
\usepackage{graphicx}
\usepackage{amssymb}
\usepackage{setspace}
\usepackage{cite}

\newtheorem{clm}{Claim}
\newtheorem{cor}{Corollary}

\begin{document}
%
%
%

\title{Space Lower Bounds for Graph Stream Problems}

\author{Paritosh Verma \inst{1}}

\institute{Birla Institute of Technology \& Science, Pilani, India \footnote[1]{This work was done as a part of author's undergradute thesis}, \email{paritoshverma97@gmail.com}}
\maketitle

\begin{abstract}
	This work concerns with proving space lower bounds for graph problems in the streaming model. It is known that computing the length of shortest path between two nodes in the streaming model requires $\Omega(n)$ space, where $n$ is the number of nodes in the graph. We study the problem of finding the depth of a given node in a rooted tree in the streaming model. For this problem we prove a tight single pass space lower bound and a multipass space lower bound. As this is a special case of computing shortest paths on graphs, the above lower bounds also apply to the shortest path problem in the streaming model. The results are obtained by using known communication complexity lower bounds or by constructing hard instances for the problem.

Additionally, we apply the techniques used in proving the above lower bound results to prove space lower bounds (single and multipass) for other graph problems like finding min $s-t$ cut, detecting negative weight cycle and finding whether two nodes lie in the same strongly connected component.
\end{abstract}

\section{Introduction}
Streaming model is a computation model in which data arrives in the form of a stream and limited memory is available to the algorithm for storing data and performing computation. Due to this memory contraint, unlike the more common Random Access Memory (RAM) model, random access of input is not allowed in the streaming model. 

This study concerns with space lower bounds for streaming problems which take graphs as input. Graph streaming algorithms find their use in applications where the size of input graph is too large to be stored in a single machine or where the input data naturally arrives in an order for instance, network packets arriving in a router. The settings where the input graph is dynamic can also be modeled using streaming algorithms by considering each update in the graph as a new data element arriving in the stream. Further, studying graph streaming algorithms also yields insights into complexity of stream computation ~\cite{mcg_survy}.

Communication complexity is a concept from information theory that is used as a tool in proving many space lower bound results in the streaming model ~\cite{rgarden_cc,amitc_sa,gurus_super_lb}. Such lower bound proofs curcially rely on the idea that the computation happening as part of a streaming algorithm can be viewed as communication happening between various (hypothetical) agents who have different chunks of input data with them. This connection between communication complexity and streaming computation is explained in more detail in the subsequent sections.

Communication problems like index problem, set disjointness problem, pointer chasing have been used to prove many known lower bounds for various streaming problems ~\cite{rgarden_cc,amitc_sa,gurus_super_lb}. In this study we use known communication complexity lower bound results to prove space lower bounds for graph streaming problems. However some of our lower bound results are based on a technique that does not rely on communication complexity.

\textbf{Our contributions. } We consider the following graph problems and prove space lower bounds for them in the streaming model-

\noindent
(i) We prove a $\Omega(n.\log{n})$ single pass and a $\Omega((n/p^7) - \log(n/p))$ $p$-pass space lower bound for the problem of finding the depth of a given node in a rooted tree.

\noindent
(ii) A $\Omega(n^2)$ single pass lower bound for the problem of computing min $s-t$ cut and the problem of detecting negative weight cycle in a weighted graph.

\noindent
(iii) A $p$-pass space lower bound of $\Omega(n^{1+ \Theta(1/p)}/poly(p))$ for the problem of finding min $s-t$ cut, finding whether two nodes lie in the same strongly connected component and the problem of detecting a negative weight cycle in a weighted graph.

In section 2, we define all the communication problems and their known lower bound results that are used in this work. Most of our multipass space lower bound proofs use the communication complexity lower bound results proved in ~\cite{gurus_super_lb}.\\ In section 3, we prove single and multi pass space lower bounds for the problem of finding depth of a node in a tree being streamed. The motivation to study this problem is that it is a simpler version of the general shortest path problem in graphs and study of this problem can yield insights into the latter problem. We prove a $\Omega(n.\log{n})$ lower bound for the single pass version of the problem which also applies to the shortest path problem. \\ 
In section 4 and 5, we use the ideas used in the above proofs to obtain single and multi pass lower bounds for the problem of finding min $s$-$t$ cuts in a graph and detecting negative weight cycles.

\subsection{Streaming Model and Communication Complexity}
In the streaming model the data is presented in the form of a stream i.e.\ data arrives in an order and random access on the input is not permitted. The space available for the algorithm is also limited. This model is useful in modelling scenarios in big data processing and cloud computing.\\
A $p$-pass (or multi pass) streaming algorithm refers to an algorithm to which input is streamed $p$ times (or many number of times). If the input is streamed only once, the streaming algorithm is called single pass.\\\\
Communication complexity is a concept from Information theory that has been used to prove space lower bounds for many problems in streaming model.\\ Communication complexity of a function $f$ is defined as the worst case communication (over all inputs) required by the best communication protocol for the following communication problem - the input of the function $f$ is divided amongst different entities which can communicate only through a channel and the goal is to compute the value of the function $f$ using minimum communication ~\cite{rgarden_cc,amitc_sa}. Next subsection describes a specific communication complexity model.


\subsubsection{Yao's communication model}

	This model consists of two players Alice and Bob (can be $n$ players in the general case). Alice and bob can communicate to each other via a channel. Alice has a binary string $x \in \{0,1\}^n$ and Bob has a binary string $y \in \{0,1\}^n$ such that both the players are unaware of the other person's string.\\
	
	Both of them are interested in computing the value $f(x,y)$ where $f$ is function of strings of both the players ~\cite{rgarden_cc}. Both of them can apriori agree on a communication protocol which they will use in order to compute the value $f(x,y)$. A trivial protocol could be that Alice sends her input $x$ to Bob via the communication channel and Bob upon receiving $x$ computes $f(x,y)$ which he then passes on to Alice. Both the players Alice and Bob are assumed to be computationally unbounded.\\
	Communication complexity of a function $f(x,y)$, $CC(f)$  is the minimum amount of bits required to be transferred through the channel by any communication protocol (for computing $f(x,y)$) in the worst case. The communication complexity of a function is in general difficult to compute because the first quantifier in it's definition is over all possible protocols.\\

  \[CC(f) = \min_{\forall \text{ protocols } P} \max_{\forall x,y } \text{(bits communicated to compute }f(x,y)\text{ using }P)\]
	\\
Most space lower bound proof involving communication complexity use the following idea:
\begin{itemize}
	\item For a given streaming problem $\mathcal{S}$ identify a underlying communication problem $\mathcal{C}$ i.e.\ a communication problem that can be reduced to the given streaming problem. Which means that using an algorithm $\mathcal{A}$ for $\mathcal{S}$ one should be able to construct a protocol for the communication problem $\mathcal{C}$.
	\item Prove a communication complexity lower bound for the communication problem $\mathcal{C}$.
	\item Translate the communication complexity lower bound for $\mathcal{C}$ to space lower bound for streaming problem $\mathcal{P}$. This step is based on the construction of the reduction.
\end{itemize}
The above idea is central to most space lower bounds results for streaming algorithms ~\cite{amitc_sa,rgarden_cc,nisan_rounds_in_cc,gurus_super_lb,zelk_mincut}.  
The following communication problems and their known communication complexity lower bounds are used in this work.


\subsection{Index problem}
In this problem there are two players Alice and Bob. Alice has an array $A$, $A\in \{ 0,1 \} ^n$ and Bob has $i \in [n]$  (where $[n]$ represents $\{0,1\}^n$). Both Alice and Bob are not aware of the other player's input and one way communication from Alice to Bob is allowed through a channel. However, Bob is not allowed to communicate to Alice. Bob wants to find out the value stored at the $i^{th}$ position of Alice's string, $A_i$. It is known that the one way communication complexity of the index problem is $ \Omega(n)$.
~\cite{rgarden_cc}.

\subsubsection{Pointer and Set Chasing}
A $(p,r)$-communication problem is defined as a communication problem which consists of $p$ players $P_1, P_2 \dots P_p$. Players communicate for $r$ rounds and are constrained to speak in the order $P_1 \rightarrow P_2 \rightarrow \dots$ to $P_p \rightarrow P_1 \rightarrow P_2 $ and so on. In the last round the player $P_p$ has to output the required value to be computed.\\
If $f:[n]\rightarrow 2^{[n]}$ be a function mapping the set $[n] = \{1, 2, 3 \dots n\}$ to $2^{[n]}$ (power set of $[n]$), then a function $f^\prime: 2^{[n]}\rightarrow 2^{[n]}$ can be defined using $f$ as follows ~\cite{gurus_super_lb}-
\[f^\prime(S) = \bigcup_{i \in s} f(i) \]
\subsubsection{Pointer Chasing}
Pointer chasing problem $PC_{n,p}$ for positive integers $n$ and $p$, is defined as a $(p,p-1)$-communication problem where $\forall i \in [p]$ player $P_i$ has a function $f_i:[n]\rightarrow [n]$. They are interested in computing the value $f_1(f_2(f_3(\dots f_p(1))))$~\cite{nisan_rounds_in_cc,gurus_super_lb}.

\begin{theorem}
	Any randomized communication protocol that solves the pointer chasing problem $PC_{n,p}$ with error probability at most $1/10$ must require at least $\Omega(n/p^4 - p^2\log(n))$ bits of communication. ~\cite{nisan_rounds_in_cc} 
\end{theorem}
\subsubsection{Set Chasing Intersection problem}
The set chasing problem $SC_{n,p}$ is defined similarly, as a $(p,p-1)$-communication problem where the $i^{th}$ player $P_i$ has a function $f_i:[n]\rightarrow 2^{[n]}$ and they are interested in computing the value ${f^\prime}_1({f^\prime}_2({f^\prime}_3 \dots {f^\prime}_p(\{1\})))$.\\


For pointer chasing and set chasing problem become easy once the number of rounds are increased from $p-1$ to $p$ or if the order of communication is inverted.\\
The set chasing intersection problem $INTERSECT(SC_{n,p})$ is a $(2p,p-1)$-communication problem in which the first $p$ players have one instance of the set chasing problem and the other $p$ players have another instance of the set chasing problem. In all the $p-1$ rounds they can communicate only in the order $P_1 \rightarrow P_2 \rightarrow P_3 \rightarrow P_4 \dots \rightarrow P_{2p} \rightarrow P_1$ and so on. Finally the goal is to check whether ${f^\prime}_1({f^\prime}_2({f^\prime}_3 \dots {f^\prime}_p(\{1\}))) \cap {f^\prime}_{p+1}({f^\prime}_{p+2}({f^\prime}_{p+3} \dots {f^\prime}_{2p}(\{1\}))) = \phi? $, in other words they are interested in knowing whether the output of the two set chasing instances intersect or not~\cite{gurus_super_lb}.\\
The following communication complexity lower bound on the set intersection problem has been proved-
\begin{theorem}
	For some positive constant $p$ such that $1 < p \leq \log(n)/(\log(\log(n)))$, any randomized communication protocol that solves the $INTERSECT(SC_{n,p})$ problem with a probability greater than $9/10$ requires $\Omega \big(n^{1+1/2(p+1)}/(p^{16}.\log^{3/2}n) \big)$ ~\cite{gurus_super_lb}.  
\end{theorem}
Using this result, space lower bounds for multi pass streaming algorithm for the problems like finding perfect matching, checking if there exists a directed path between two vertices has been proved in~\cite{gurus_super_lb}.

\section{Finding Depth of a node in a Tree}
The problem that is considered here is is the following- Let $\mathcal{T}$ be a rooted tree whose root is denoted by a known symbol $r$ and $u$ is some node in the tree. Given a stream $\sigma$ consisting of the node $u$ followed by edges of the tree $\mathcal{T}$, the problem is to compute the depth of the node $u$ in $\mathcal{T}$.\\
This problem is a simpler version of the general problem of finding distance between two nodes in a graph.
\subsection{Multipass Lower Bound}
Using the communication complexity lower bound of the pointer chasing problem we prove the following multi pass space lower bound for the mentioned problem.

\begin{clm}
	Any randomized $p$-pass streaming algorithm that computes the depth of a given node in a rooted tree with error probability at most $1/10$ must require at least $\Omega((n/p^7) - \log(n/p))$ space. 
\end{clm}

\begin{proof}
	Given a $p$-pass streaming algorithm $A$ for finding the depth of a node in a rooted tree that uses $s$ bits of space we design a communication protocol that solves the pointer chasing problem $PC_{n,p+1}$ using $s.\Theta(p^2)$ bits of communication (through channel between players).\\
	    The construction of the communication protocol is based on the idea of visualizing the computation performed in the pointer chasing problem as a graph. The function $f_i$ of each player $P_i$, $\forall i \in [p+1]$ can be visualized as bipartite graphs. As a result, the composition of the functions $f_i$ can be viewed as side by side concatenation of these bipartite graphs as shown in figure 1. In this view, the goal of the pointer chasing problem is to find the node to which the bold edges emerging from the node $n_1$ lead. The node to which the bold edges lead corresponds to the value $f_1(f_2(f_3(\dots f_{p+1}(1))))$.\\ 
	    
		Given the algorithm $A$ in each of the $p$ rounds all the $p+1$ players will stream the edges of the bipartite graph (corresponding to their function $f_i$) to $A$ and then they will pass the memory transcript of $A$ to the next player according to the communication order defined in pointer chasing problem. At the end of each round player $P_{p+1}$ will stream the edges shown in blue (see figure 1) in addition to the edges corresponding to $f_{p+1}$. All the edges that are streamed to $A$ can be viewed as a tree rooted at node $N$ (see lemma 1 below).
			This means that after $p$ round of communication, the value of $f_1(f_2(f_3(\dots f_{p+1}(1))))$ can be found out by finding the depth of the node $n_1$ (see figure 3) in the computation graph/tree. Let $d(v)$ denote the depth of a node $v$ and $x$ be the node corresponding to .$f_1(f_2(f_3(\dots f_{p+1}(1))))$. Then, $d(x) = d(n_1) - 1 -p$ . The depth of node $n_1$ is the output algorithm $A$ gives at the end of the last round. Hence by knowing the value of $d(n_1)$ we can determined $d(x)$ thereby determining the node $x$ itself. This is because each of the rightmost node (figure 1) has different depth with respect to root $N$.\\
	The communication complexity of this protocol is $s.\Theta(p^2)$ as the size of the memory transcript is $s$ and order of the total messages sent between players is $\Theta(p^2)$ ($\Theta(p)$ messages in $p$ rounds). According to the communication complexity lower bound of the pointer chasing problem
					\[s.\Theta(p^2) = \Omega \big(\frac{(n/p)}{p^4} - p^2\log(n/p)\big)\]
						\[s = \Omega\big(n/p^7 - \log(n/p)\big)\]
		Where $n$ is the total number of nodes in the computation graph and $n/p$ represents the domain of functions $f_i$.
\end{proof}

\begin{figure}
	\centering
	\includegraphics[width = 10cm, height=7.5cm]{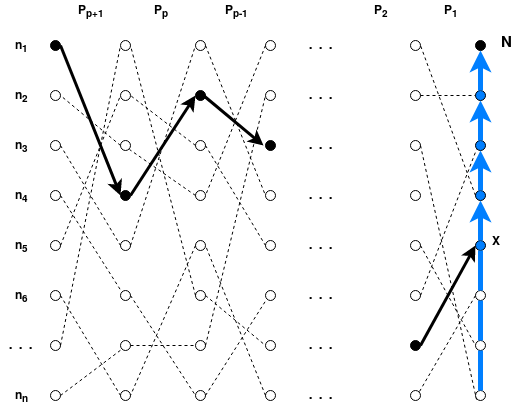}
	\caption{Computation graph $G$ for $PC_{n,p+1}$, with blue edges added for node depth lower bound}
\end{figure}

\begin{lemma}
	The computation graph $G$ for the pointer chasing problem is a tree.
\end{lemma}
\begin{proof}
	Let's assume that $G$ has a cycle $\mathcal{C}$. Each edge of the cycle C can be classified into either a blue edge of an edge corresponding to some function $f_i$, $\forall i \in [p+1]$. Also, the cycle cannot be completely composed of blue edges. Due to this, we can find the minimum $i$, such that edge corresponding to $f_i$ is in $\mathcal{C}$. Call this $i^{\prime}$. Since $\mathcal{C}$ is a cycle, there will be two edges corresponding to function $f_{i^{\prime}}$, which is a contradiction since $f_{i^{\prime}}$ is a function. Hence no cycle can exist.
\end{proof}
This result also shows that any one pass streaming algorithm must require at least $\Omega(n)$ space to compute the depth of a given node in a tree being streamed (for $p=1$). In the next section we prove a stronger $\Omega(n.\log(n))$ lower bound for the single pass version of the problem.
\subsection{Stronger bound for the single pass algorithms}
	In this section we prove a stronger space lower bound for the single pass version of the problem. The given lower bound proof does not follow the standard reduction procedure that is used to prove most space lower bound results in streaming model.
\begin{clm}
	Any one pass streaming algorithm that computes the depth of a given node in a rooted tree must require $\Omega(n.log(n))$ space.
\end{clm}

\begin{proof}
		To establish this result we first choose a particular input instance and then we prove a lower bound on the space required by the best algorithm for that instance, this value according to yao's minimax principle serves as a space lower bound for the problem.\\
			The input instance we consider is the instance in which the node $v$ (whose depth is to be calculated) is a leaf node and the edge connecting $v$ to its parent in the tree $(v,p(v))$ arrives as the last edge of the stream. The idea here is that the essential information required to compute the depth of the node $v$ i.e.\ the location of the node $v$ in the tree is deferred in the input stream.\\
				Let $A$ be some algorithm that computes the depth of a given node. Now consider the memory transcript $\mathcal{M}$ of the algorithm $A$ when it has processed all the edges except the last edge $(v,p(v))$ of the constructed input instance. Now we claim that $|\mathcal{M}| = \Omega(n.\log(n))$, proving this is sufficient to prove the claim.\\
					To prove $|\mathcal{M}| = \Omega(n.\log(n))$ we argue that given the memory transcript $\mathcal{M}$ and the algorithm $A$ one can recover the depths of all nodes in $V\setminus \{v\}$ using $2|\mathcal{M}|$ space. Given $\mathcal{M}$ and $A$ one can run the algorithm $n-1$ times, each time continuing the computation on $\mathcal{M}$ and streaming the edge $(v,u)$ $\forall u \in V\setminus\{v\}$. From all of these runs we can recover the depths of all the nodes in the set $V\setminus\{v\}$, as depth of node $u$ is one less than depth of node $v$ and $u \in V\setminus\{v\}$. To complete the proof we show that the depth information of all the nodes of a tree requires $\Theta(n.\log(n))$ bits to store which implies that $|\mathcal{M}| = \Omega(n.\log(n))$.\\
						Let $\mathcal{D}$ denote the number of different functions $d:V\rightarrow \mathbb{Z}^{+}$ such that $d(v)$ represents the depth of the node $u \in V\setminus\{v\}$. \\
							Let $\mathcal{D}_i$ denote the number of different functions $d:V\rightarrow \mathbb{Z}^{+}$ such that $d(v)$ represents the depth of the node and $\max_{v \in V}d(v)  = i$ and $S(n,k)$ is stirling number of second kind.
							
								\[|\mathcal{M}| = \Omega(\log_2(\mathcal{D}))\]
									\[|\mathcal{M}| = \Omega(\log_2(\sum_{i \in [n-2]} \mathcal{D}_i))\]
										\[|\mathcal{M}| = \Omega(\log_2(\sum_{i \in [n-2]} n.S(n-1,i)))\] 
											\[|\mathcal{M}| = \Omega(\log_2(n.\sum_{i \in [n-2]} i^{n-i}))\] 
												\[|\mathcal{M}| = \Omega(\log_2(n.(n/2)^{n/2})) \]     
													\[|\mathcal{M}| = \Omega(n.\log_2(n))\]
													
\end{proof}
Note that, the lower bound proved above is tight because by using $\Theta(n.\log_2(n))$ space we can essentially store the entire tree in the working memory.

\begin{cor}
	Any streaming one pass algorithm that computes the distance between two nodes in a graph requires at least $\Omega(n.\log_2(n))$ space.
\end{cor}
This above corollary is applicable even for the shortest path problem, that is any streaming algorithm that finds the distance between two given vertices of a graph must require $\Omega(n.\log{n})$ space.\\
Using the same idea of deferring the essential input in the stream, we prove a space lower bound for the problem of min s-t cut in weighted graph.

\section{Min $s-t$ Cut Problem}
To prove this lower bound we use the following known lower bound for the unweighted min cut problem.\\
\subsection{Single Pass Lower Bound}
\begin{theorem}
	Any one pass streaming algorithm that computes the min cut of a unweighted graph requires at least $\Omega(n^2)$ memory. ~\cite{zelk_mincut}
\end{theorem}
The input for this problem is a weighted graph stream i.e. the individual tokens of the stream are weighted edges of the form $((u,v),w)$ where $w$ is the weight of the edge $(u,v)$.
\begin{clm}
	Any one pass streaming algorithm that computes the min $s$-$t$ cut value for given pair of nodes $s$ and $t$ in a weighted graph must require at least $\Omega(n^2)$ space. 
\end{clm}
\begin{proof}

	Let $A$ be any streaming algorithm for the weighted min $s$-$t$ cut problem and the space required by the algorithm be $s$. Now we use this algorithm $A$ to construct a streaming algorithm $A_s$ for solving the unweighted min cut problem that uses $\Theta(s)$ space. Since $A_s$ requires $\Theta(s)$ space, using theorem 3 implies that $s = \Omega(n^2)$, which completes the proof.\\ 
	Let $G=(V,E)$ be the graph whose min cut is to be calculated, then we construct another weighted graph $G^\prime = (V\cup\{u,v\},E\cup \{(u,x),(v,y)\})$, for $x,y \in V$ such that weight of the edges $(u,x)$ and $(v,y)$ is $n$ and every other edge weighs one. These weights ensure that the edges $(u,x)$ and $(v,y)$ never lie in the min $u$-$v$ cut of $G^\prime$. This implies that min $u$-$v$ cut of graph $G^\prime$ is same as the min $x$-$y$ cut of graph $G$, $G$ being an unweighted graph. \\Now to compute the min cut for $G$, the instance $G^\prime$ is streamed to the algorithm $A$ multiple times, one for each $y \in V\setminus \{x\}$, for a fixed value of $x$. The graph $G^\prime$ is streamed in such a way that the edge $(v,y)$ arrives as the last edge of the stream. The min cut value of graph $G$ is calculated by taking the min value of the all the min $u$-$v$ cut values calculated for the graph $G^\prime$ (same as min $x$-$y$ cut for $G$).\\
	All the $n-1$ min cut values (corresponding to $y \in V\setminus \{x\}$) that are required to be computed can be computed using $2s$ space as the memory transcript $\mathcal{M}$ of the algorithm $A$ - after it has processed all the edges $E\cup \{(u,x)\}$, can be reused to compute all the min $u$-$v$ ($x$-$y$) cuts $\forall y \in V\setminus \{x\}$.\\ 
	To calculate the value of min $x$-$y'$ cut for some $y' \in V \setminus \{x\}$, the memory transcript $\mathcal{M}$ is used along with the algorithm $A$ and the edge $(y',v)$ is streamed to $A$, then the value computed by the algorithm $A$ is the min $u$-$v$ cut for $G^\prime$ (or min $x$-$y'$ cut for $G$). The memory transcript can be copied and used similarly $n-1$ times for each $y' \in V \setminus \{x\}$, using $2s$ space.\\
	This leads to a $2s$ space streaming algorithm that computes min cut of a unweighted graph being streamed, it implies that $s = \Omega(n^2)$. Which means that for all possible algorithms $A$, the space required must be $\Omega(n^2)$.
\end{proof}

In the next section, using the known communication complexity lower bound for the set chasing intersection problem we prove a multi pass space lower bound for the unweighted min cut problem.

\begin{figure}
	\centering
	\includegraphics[width = 12.5cm, height=6.5cm]{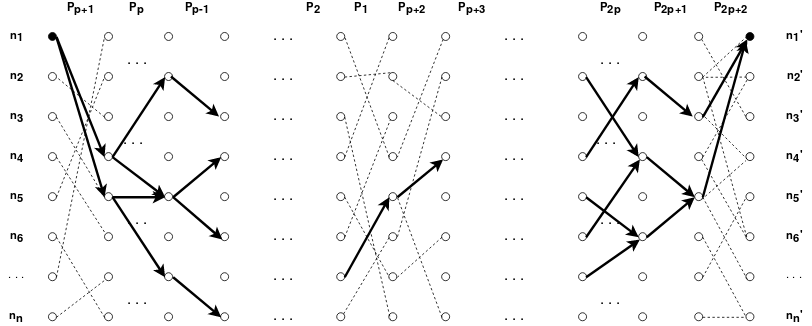}
	\caption{Computation graph for $INTERSECT(SC_{n,p+1})$, directions are used to indicated paths emerging and leading to nodes $n_1$ and $n_1^{'}$. All the edges corresponding to functions $f_{i}^{\prime}$ are not shown for clarity.}
\end{figure}

\subsection{Multi Pass Lower Bound}
\begin{clm}
	Any $p$ pass streaming algorithm that computes the min $s$-$t$ cut of a unweighted graph must require $\Omega \big(n^{1+1/2(p+1)}/(p^{19}.\log^{3/2}n) \big)$ bits of space.
\end{clm}
\begin{proof}
	Suppose there exists such an algorithm $A$ that uses at most $s$ bits of memory. Then we can use this algorithm to design a protocol for the set chasing intersection problem $INTERSECT(SC_{n,p+1})$ that uses $s.\Theta(p^2)$ bits of communication.\\
		The computation performed in the set chasing intersection problem can be visualized as a computation graph, shown in figure 2. For all the $2(p+1)$ players their functions $f_{i}^{\prime}$ can be viewed as bipartite graphs. \\
		We use the algorithm $A$ to compute the min $n_1$-$n_1^{\prime}$ cut of the computation graph $G$ (shown in figure 2). This is done as follows. Each player $P_i$ uses the algorithm $A$ and streams the edges corresponding to their function $f_i$ into the algorithm (as the individual functions can be visualized as a bipartite graph) and passes the memory transcript to the next player according to the communication order constraint. In this view, any path going from node $n_1$ to $n_1^{\prime}$ would represent an element in the set ${f^\prime}_1({f^\prime}_2({f^\prime}_3 \dots {f^\prime}_{p+1}(\{1\}))) \cap {f^\prime}_{p+2}({f^\prime}_{p+3}({f^\prime}_{p+4} \dots {f^\prime}_{2p+2}(\{1\})))$.\\
			Now we can claim that by knowing the value of min $n_1$-$n_1^{\prime}$ cut we can find out whether the two instances of set chasing intersect. This is because if the two instance of set chasing do not intersect then the size of min $n_1$-$n_1^{\prime}$ cut would be zero as there is not path from the vertex $n_1$ to $n_1^{\prime}$. If on the other hand the outputs of the two set chasing instances intersect then the min cut value would be greater than zero as the intersection would yield a path from $n_1$ to $n_1^{\prime}$. Hence by checking whether the min cut value is zero or not one can solve the set chasing intersection problem. Using the communication complexity lower bound- 
				\[s.\Theta(p^2) = \Omega \Big(\frac{(n/p)^{1 + \Theta(1/p)} }{p^{16}.\log^{3/2}n}\Big) \]
					\[s = \Omega \Big(\frac{ n^{1 + \Theta (1/p) } }{p^{19}.\log^{3/2}n}\Big) \]
Where $s.\Theta(p^2)$ is the total communication required during the protocol.
\end{proof}
In the next section we study the problem of detecting negative weight cycles in a graph stream. We prove single and multipass lower bounds for the problem using reductions from index problem and set chasing problem respectively.

\section{Detecting Negative Weight Cycle}
\subsection{Single Pass Lower Bound}
Using the known communication complexity lower bound for the index problem we first prove a single pass space lower bound for detecting negative weight cycle.
\begin{clm}
	Any streaming algorithm that can detect the presence of a negative weight cycle in a weighted graph stream must use at least $\Omega(n^2)$ space. 
\end{clm}
\begin{proof}
	Let $C$ be a streaming algorithm that detects the presence of a negative weight cycle using $s$ space. Then $C$ can be used to design a communication protocol for the index protocol as follows.\\
	Let $(A,i)$ be an instance of index problem in which $A$ is of the size $\Theta(n^2)$, then the binary string $A$ can be interpreted as a graph on $n$ vertices. Alice can stream the edges of this graph to the algorithm $C$ associating with each edge a weight of positive one unit. Then she can send the memory transcript obtained (of size $s$) to Bob. \\
	Let $(a,b)$ be the edge corresponding to Bob's input $i$. After receiving the memory transcript Bob streams the edges $(a,v)$ and $(b,v)$ to the streaming algorithm associating with each edge a weight of negative one. \\
	Bob can now find the $i^{th}$ bit of Alice's string by knowing whether a negative weight cycle is present or not. This is due to the fact that if edge $(a,b)$ is present in the graph, the vertices $a,b$ and $v$ form a cycle of weight negative one. If the edge $(a,b)$ is not present in the graph then the minimum length of the cycle containing the negative weight edges is $4$ which means that it's weight will be non negative, hence negative weight cycle will not exist. Negative weight cycle exists if and only if edge $(a,b)$ is present in the graph i.e. when $A_i = 1$.\\
	This leads to a $s$ bit one way communication protocol that solves the index problem. It means that according to the lower bound result of Index problem, $s = \Omega(n^2)$. 
\end{proof}

\begin{figure}
	\centering
	\includegraphics[width = 12.5cm, height=6.5cm]{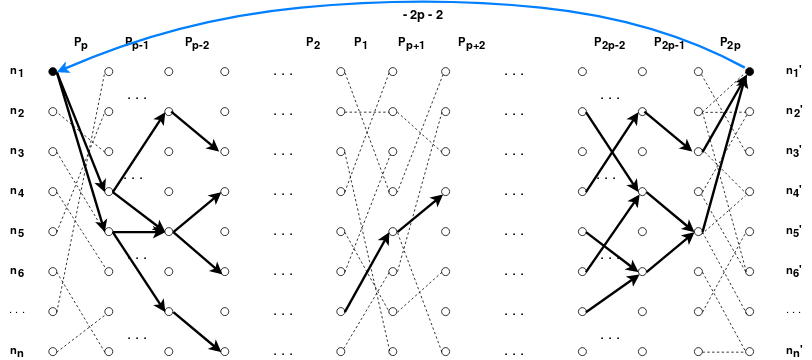}
	\caption{Computation graph for $INTERSECT(SC_{n,p+1})$, blue edge is added by the player $P_{2p}$ in every round.}
\end{figure}

\subsection{Multi Pass Lower Bound}
A multi pass lower bound for the same problem can be proved using the communication complexity lower bound for the set chasing intersection problem.\\

\begin{clm}
	 Any $p$ pass streaming algorithm that detects the presence of a negative weight cycle in a weighted graph being streamed must require at least $\Omega \big(n^{1+1/2(p+1)}/(p^{19}.\log^{3/2}n) \big)$ bits of space. 
\end{clm}
\begin{proof}
	Suppose there exists a $p$-pass streaming algorithm $A$ for the problem of detecting a negative weight cycle which uses at most $s$ bits of space.	Then such an algorithm can be used to design a protocol for the set chasing intersection problem $INTERSECT(SC_{n,p+1})$ as follows-\\
	In each of the $p$ rounds all the $2(p+1)$ players stream edges the edges corresponding to their function (as shown in figure 3) and pass the resultant memory transcript of the algorithm to the next player. The weight of one is assigned to every edge being streamed. The player $P_{2(p+1)}$ also adds the edge shown in blue (see figure 3) having weight $-2(p+1)-1$ in every round.\\
	After the completion of $p$ rounds, the set chasing problem can be solved by asking whether the graph that is streamed to algorithm $A$ has a negative weight cycle or not. This is because there exists a negative weight cycle in the graph if and only if the output sets intersect in the set chasing problem. As if the output set does not intersect, there is no path from node $n_1$ to $n_1^{\prime}$ which means no negative weight cycle. If the output sets intersect, we have a cycle of weigh $-1$, consisting of the blue edge (weight $-2(p+1)-1$) and a path of $2(p+1)$ edges (from $n_1$ to $n_1^{\prime}$) of weight $1$ each.\\
	Since the size of the memory transcript is $s$ bits and in each of the $p$ rounds the memory transcript is transferred $\Theta(p)$ times, the total communication done by the protocol is $s.\Theta(p^2)$ which according to the lower bound on communication complexity of the set chasing intersection problem means that-
					\[s.\Theta(p^2) = \Omega \Big(\frac{ (n/p)^{1 + \Theta(1/p)}}{p^{16}.\log^{3/2}n}\Big) \]
					\[s = \Omega \Big(\frac{ n^{1 + \Theta (1/p) } }{p^{19}.\log^{3/2}n}\Big) \]
								since $p^{\Theta(1/p)} = O(1)$.  
									
\end{proof}
The exact same proof can also be extended to prove same lower bound result for the problem of finding whether two nodes lie in the same strongly connected component or not in a directed graph being streamed.

\bibliographystyle{splncs04}
\bibliography{master}
\end{document}